 \documentclass[final,3p,times]{elsarticle}

\usepackage{amsmath,amssymb,amsfonts}
\usepackage{color}
 \usepackage{amsthm}

\newcommand{\owa}{\mbox{OWA}}

\newcommand{\RP}{\mathbb{R}^p}
\newcommand{\R}{\mathbb{R}}

\newcommand{\vij}{v_{i}^{j}}
\newcommand{\ej}{e^{j}}
\newcommand{\Sol}{\mathcal{T}}
\newcommand{\blue}{\mbox{\tt blue}}
\newcommand{\red}{\mbox{\tt red}}

\newtheorem{prop}{\bf Proposition}
\newenvironment{pf}[1][]{{\noindent \bf Proof. #1 }}{\hfill $\Box$\\ }

\newtheorem{defi}{Definition}

\newtheorem{example}{Example}

\def\keyw#1{{\bf#1}}
\def\tab{\hspace*{5mm}}

\newsavebox{\fminibox}
\newlength{\fminilength}

\newenvironment{fminipage}[1][\linewidth]%
{\vspace{\parsep}
    \setlength{\fminilength}{#1}%
    \setlength{\fboxrule}{0.5pt}
    \setlength{\fboxsep}{1.5mm}
    \addtolength{\fminilength}{-2\fboxsep}%
    \addtolength{\fminilength}{-2\fboxrule}%
    \begin{lrbox}{\fminibox}\begin{minipage}{\fminilength}}
  {\end{minipage}\end{lrbox}\noindent\fbox{\usebox{\fminibox}}\vspace{\parsep}}

\def\algo#1#2#3{\begin{fminipage}{{\bf Algorithm~}{\sc#1}\\{\bf Input~:~}#2\\{\bf Output~:~}#3\\}}
\def\endalgo{\end{fminipage}}

\journal{Computers \& Operations Research}

\begin{document}

\begin{frontmatter}




\title{Exact algorithms for $\owa$-optimization in multiobjective\\ spanning tree problems}
\author[LIP6]{Lucie Galand}
\ead{lucie.galand@lip6.fr}
\author[LIP6]{Olivier Spanjaard\corref{cor1}}
\ead{olivier.spanjaard@lip6.fr}
\address[LIP6]{LIP6-CNRS, Universit\'e Pierre et Marie Curie (UPMC), 104 Avenue du Pr\'esident Kennedy, F-75016 Paris, France}
\cortext[cor1]{Corresponding author.}

\begin{abstract}
This paper deals with the multiobjective version of the optimal
spanning tree problem. More precisely, we are interested in determining the optimal
spanning tree according to an \emph{Ordered Weighted Average} (OWA) of its objective values. We first show that the problem is
weakly NP-hard. In the case where the weights of the OWA are strictly decreasing, we then propose a mixed integer programming formulation, and provide dedicated optimality conditions yielding an important reduction of the size of the program. Next, we present two bounds that can be used to prune
subspaces of solutions either in a shaving phase or in a branch and bound procedure. The validity of these bounds does not depend on specific properties of the weights (apart from non-negativity). All these exact resolution algorithms are compared on the basis of numerical experiments, according to their respective validity scopes. 
\end{abstract}


\end{frontmatter}

\section{Introduction}
\label{intro}
Multiobjective combinatorial optimization deals with problems
involving multiple viewpoints \cite{Ehrgo05}. More formally, the
valuation structure of such problems is made of vectors (each
component representing a specific viewpoint) instead of scalars. All
popular single objective optimization problems (e.g., valued graph
problems, integer linear programming...) can be recasted in this
setting, and solution algorithms must be proposed. Two types of
approaches can be studied: either one looks for a best compromise
solution according to a given aggregation function (e.g., max
operator, Chebyshev's norm to a reference vector \cite{Wierz80},
ordered weighted average \cite{Yager88}, Choquet integral
\cite{Grabi96}), or one aims at generating the whole set of
Pareto-optimal solutions (i.e., solutions that cannot be improved on
one objective without being depreciated on another one). Computing this set (also called
the Pareto set) is natural when no preferential information is available 
or when the available information is unsufficient to elicit the parameters of the aggregation function: the solutions in the Pareto
set are indeed the only ones likely to be selected by any rational
decision maker. The interest in this approach has spawned a
substantial literature (for a survey on the topic, the reader can
refer to several quite recent papers \cite{EhrgG00,EhrgG04}). However,
the number of Pareto-optimal solutions can grow exponentially
with the size of the instance \cite[e.g.][]{HamaR94,Hanse80} and the number
of objectives \cite{Rosin91}. The
examination of the whole set of Pareto-optimal solutions can therefore become
quickly awkward. Furthermore, focusing on one particular compromise solution
makes it possible to considerably speed up the resolution
procedure. Consequently, for practicality and
efficiency reasons, the search for a best compromise solution seems to be a
better approach when an aggregation function is available. This is the approach we study in this paper. 

More precisely, we investigate here the multiobjective version of the
optimal spanning tree problem. This problem arises naturally in
various contexts. For example, consider a broadcasting network where
the values of the edges represent bandwidths. Assuming that the
bandwidth of a chain equals the minimum bandwidth over its edges, it
is well-known that, in a maximal spanning tree, the bandwidth between
two nodes is the maximum possible. When there are several scenarios of
traffic (impacting the values of the bandwidths) or several opinions
of experts on the values of the bandwidths, the problem becomes
multiobjective. Previous works on the subject mainly deal with generating the whole Pareto set in the biobjective case \cite{AndJL96,HamaR94,SourS08,SteiR08}, or computing a min-max (regret) optimal solution when there are more than two objectives \cite{AisBV09,HamaR94,KouvY97,Warbu85,Yu98}. To our knowledge, there is therefore no operational algorithmic tool for this problem when there are more than two objectives and when the min-max (regret) criterion is not really suitable. The present paper precisely aims at tackling this gap, by providing algorithms able to optimize a less conservative decision criterion for any number of objectives. Provided the required preferential information is available, and provided the objectives are commensurate (which is the case in the above example for instance), we propose to resort to an averaging operator to compare
the vectorial values of the feasible solutions (spanning trees). According to the decision context, one may however want to put the emphasis on the best, the worst or the median evaluations of a solution. In other words, one needs to assign
importance weights not to specific objectives (scenarios, experts), but rather to
best and worst evaluations. The \emph{Ordered Weighted Average} (OWA)
precisely enables to model such concern. For this reason, we focus
here on the \emph{OWA-optimal spanning tree problem}, i.e. finding the
optimal spanning tree according to OWA in a multiobjective spanning
tree problem. 
In the case of strictly decreasing weights (favouring well-balanced
solutions), the use of an $\owa$ objective function has been studied
by Ogryczak and Sliwinski \cite{Ogryc09,OgryS03} in continuous
optimization under linear constraints, and by Galand and Spanjaard \cite{GalaS07b} in multiobjective heuristic search. In this paper, we propose new algorithms specifically dedicated to the optimization of OWA in multiobjective spanning tree problems, some of which are able to handle not only strictly decreasing weights but also every other types of weights.

The paper is organized as follows. After recalling some preliminary
definitions and stating the problem, we give some insights into computational complexity (Section 2). Then, we provide a mixed integer programming formulation of the problem, as well as a preprocessing procedure based on optimality conditions for the $\owa$-optimal spanning tree problem (Section 3). We then propose two -- efficiently computable -- alternative bounds for discarding subspaces
of solutions in either a branch and bound or a shaving procedure (Section 4). Finally, we provide numerical experiments to assess the operationality of the proposed methods (Section 5).

\section{Preliminaries}
\label{sec:OWA}

\subsection{Multiobjective compromise search problem}
\label{sec:comp}

A \emph{multiobjective compromise search problem} is a
problem endowed with vectorial costs where one searches for a best
compromise solution according to a given aggregation function. A generic multiobjective compromise search problem can be formulated as a mathematical program.
We now introduce some notations for this purpose. 
We denote by $X \subseteq \{0,1\}^m$ the set of feasible
solutions, $f:X \rightarrow \RP$ a vector valued function on $X$, and $\varphi:\RP \rightarrow \R$ a multiobjective aggregation function. Within this setting, a multiobjective compromise search problem is written as follows: 

$$(P) ~~~\left\{\begin{array}{ll}
 \min ~~ \varphi(y)\\
 \mbox{s.t. } y  =  f(x)\\
 x  \in  X
\end{array}\right. $$    

Denoting by $Y= f(X) = \{f(x): x \in X\}$ the \emph{image set} of $X$ in the \emph{objective space} $\RP$, problem $P$ can be simply reformulated as $\min_{y \in Y} \varphi(y)$. The $\owa$-optimal spanning tree problem obviously belongs to this class of
problems. The resolution methods we propose hereafter for this problem are actually quite generic. 
The mathematical formulations of $P$ will therefore be convenient to describe these methods in the sequel.
We now more specifically introduce the $\owa$ operator and then the $\owa$-optimal spanning tree problem.

\subsection{The $\owa$ operator}
\label{sec:def}

Given a set $\{1,\ldots,p\}$ of objectives (to minimize), one can associate a vector in ${\mathbb N}^p$ to every feasible solution of a multiobjective problem. The comparison of solutions amounts then to comparing the corresponding vectors. Following several works in multiobjective optimization \cite[e.g.][]{Ogryc00,PernS03,PerSS06}, we propose to compare the vectors on the basis of their OWA value \cite{Yager88} (to minimize), defined as follows:

\begin{defi}
Given a vector $y$ in ${\mathbb R}^p$, its ordered weighted average is $\owa(y)$=$\sum_{i=1}^p w_i y_{(i)}$, where $\sum_{i=1}^p w_i = 1$ and $y_{(1)}$ $\geq$ .. $\geq$ $y_{(p)}$ are the components of $y$ sorted in non-increasing order.
\end{defi}

According to the decisional context, one can distinguish within the class of OWA operators two interesting subclasses depending on the definition of the weights:
\begin{itemize}
\item \emph{strictly decreasing weights}, i.e. $w_1 > \ldots > w_p >
  0$: the set of weights naturally belongs to this class when
  performing robust discrete optimization. In \emph{robust
    optimization} \cite{KouvY97}, the cost of a solution depend on
  different possible scenarios (states of the world). The aim is to
  find a \emph{robust solution} according to this multiobjective
  representation, i.e. a solution that remains suitable whatever
  scenario finally occurs. The use of the $\owa$ criterion in this setting is justified since it can be characterized by a set of axioms that are natural for modelling robustness \cite{PernS03}. More precisely, these axioms characterize an OWA criterion with strictly positive and strictly decreasing weights. Compared to
  the $\max$ criterion frequently used in robustness, the OWA
  criterion is less conservative since it enables trade-offs between
  several scenarios. Note however that the $\owa$ criterion includes
  the $\max$ criterion as a special case, when one sets
  $w_1=1-\sum_{i=2}^p \varepsilon_i$, $w_2=\varepsilon_2$, $\ldots$,
  $w_p=\varepsilon_p$ with $\varepsilon_2 > \ldots > \varepsilon_p >0$
  and $\varepsilon_2$ tends towards $0$. Another interesting special
  case is obtained for ``big-stepped weights'', i.e. when the gaps
  between successive weights are huge ($w_1 \gg \ldots \gg w_p$). The
  $\owa$ criterion reduces then to the leximax operator, which
  consists in comparing two vectors on the basis of their greatest
  component, their second greatest one in case of equality on the
  first one, and so on... This criterion refines thus the $\max$
  criterion by discriminating between vectors with the same value on
  the greatest component.
\item \emph{non-monotonic weights}: one of the most famous decision criterion in decision under complete uncertainty (i.e., when several states of the world can occur, and no information is available about their plausibilities) is Hurwicz's criterion, that enables to model intermediate attitudes towards uncertainty (i.e., neither desperately pessimistic nor outrageously optimistic) by performing a linear combination of the maximum possible value of a solution under the different scenarios, and the minimum possible one. More formally, if $y$ is the image of a solution in the objective space, then its value according to Hurwicz's criterion is: $\alpha \max_i y_i + (1-\alpha) \min_i y_i$. Hurwicz's criterion clearly is a special case of the OWA criterion, obtained by setting $w_1=\alpha$, $w_p=1-\alpha$ and $w_i=0$ $\forall i\neq 1,p$. As soon as there are more than two scenarios and $\alpha \not\in \{0,1\}$, these sets of weights are neither non-increasing nor non-decreasing. Another natural decision context where the sequence of weights is non-monotonic happens when every component represents the opinion of a particular expert, and one wants to dismiss the extreme opinions. For instance, one can set $w_1=0$, $w_p=0$ and $w_i=1/(p-2)$ $\forall i\neq 1,p$.
\end{itemize}

\subsection{Problem and complexity}

The problem we study in this paper is the OWA-optimal spanning tree problem, that can be formulated as follows:\\ [1ex]
\textsc{$\owa$-optimal spanning tree problem ($\owa$-ST)}\\
\emph{Input: } a finite connected graph $G = (V,E)$, $p$ integer valuations $v_i^e$ for every edge $e \in E$ $(i=1,\ldots,p)$, and a set of weights $w_i$ ($i=1,\ldots,p$) for the $\owa$ criterion;\\
\emph{Goal: } we want to determine a spanning tree $T^* \in ~\mbox{arg}\;\min_{T \in {\mathcal T}}\owa(f(T))$, where ${\mathcal T}$ is the set of spanning trees in $G$ and $f(T) = (\sum_{e \in T} v_1^e,\ldots,\sum_{e \in T} v_p^e)$. \\[-1.5ex] %

Coming back to generic formulation $P$ of a multiobjective compromise search problem, problem $\owa$-ST corresponds to the following specifications: a spanning tree $T$ is characterized by a vector $x=(x^1,\ldots,x^m) \in X$ where $x^j = 1$ when edge $\ej$ belongs to $T$, and its cost is $f(x) = (f_1(x),\ldots,f_p(x))$ where $f_i(x) = \sum_{j=1}^m \vij x^j$. The aggregation function is
of course defined by $\varphi(y) = \owa(y)$. 
We now introduce a small instance of problem $\owa$-ST that will be used as a running example in the sequel of the paper.
\begin{example} \label{EXP1}
Consider a clique $G$ with 4 vertices, and assume that the edges have been evaluated according to 3 objectives (scenarios, experts...). The set of vertices is $V = \{1,2,3,4\}$ and the valuations of the edges are the following: $v^{[1,2]} = (3,2,3)$, $v^{[1,3]} = (4,3,1)$, $v^{[1,4]} = (1,2,2)$, $v^{[2,3]} = (2,4,1)$, $v^{[2,4]} = (2,6,1)$, $v^{[3,4]} = (1,5,1)$. A minimum spanning tree is $T_1 = \{[1,4],[2,3],[3,4]\}$ according to the first dimension and the arithmetic mean, $T_2 = \{[1,2],[1,3],[1,4]\}$ according to the second dimension, and $T_3 = \{[1,3],[2,3],[3,4]\}$ according to the third dimension. None of them is however completely satisfying: either it is too much unbalanced ($f(T_1) = (4,11,4)$ and $f(T_3) = (7,12,3)$), or it is too much conservative ($f(T_2) = (8,7,6)$). Spanning tree $T_4 = \{[1,2],[1,4],[2,3]\}$ (with $f(T_4) = (6,8,6)$) presents none of the drawbacks of the previous trees, and is OWA-optimal when setting for instance $w_1 = 0.5$, $w_2 = 0.3$ and $w_3 = 0.2$. The overall ranking according to OWA on $T_1, \ldots, T_4$ is indeed: 1) $T_4$  with $\owa(6,8,6) = 7$, 2) $T_2$ with $\owa(8,7,6) = 7.3$, 3) $T_1$ with $\owa(4,11,4) = 7.5$, 4) $T_3$ with $\owa(7,12,3) = 8.7$. 
\end{example}
As already indicated above, when $w_1 = 1$ and all other weights $w_i$'s are equal to zero, the $\owa$ criterion reduces to the $\max$ criterion. In this case, it has been proved that computing a $\min$-$\max$ spanning tree is NP-hard \cite{HamaR94,Yu98}. Consequently, problem $\owa$-ST is also NP-hard in the general case. Note however that it is polynomially solvable for some subclasses of instances:
\begin{itemize}
\item when $w_1 = w_2 = \ldots = w_p$ (arithmetic mean), an optimal solution can be obtained by valuing every edge $e$ by $\sum_{i=1}^p v_i^e$, and then applying a standard minimum spanning tree algorithm (e.g. Prim's algorithm or Kruskal's algorithm);
\item when $w_p=1$ and $w_i=0$ $\forall i \neq p$ ($\min$ criterion), an optimal solution can be obtained by solving $p$ standard minimum spanning tree problems (one for each objective) and then returning the optimal one among them according to $\owa$; 
\item when there exists a permutation $\pi$ of objectives such that $v_{\pi(1)}^e \geq \ldots \geq v_{\pi(p)}^e$ for every edge $e$, an optimal solution can be obtained by valuing every edge $e$ by $\sum_{i=1}^p w_i v_{\pi(i)}^e$, and then applying a standard minimum spanning tree algorithm.
\end{itemize}
Note that these types of polynomial instances are quite
uncommon. Nevertheless, one can design a \emph{Fully Polynomial Time
  Approximation Scheme} (FPTAS) that works for any instance of
$\owa$-ST. For this purpose, following Aissi et al. \cite{AisBV07},
one can enrich the existing FPTAS for the determination of an
approximate Pareto set in the multiobjective spanning tree problem
\cite{PapaY00}. This algorithm indeed returns a set ${\mathcal
  T_{\varepsilon}}$ of spanning trees, the cardinality of which is
polynomial in the size of the instance. For every spanning tree $T$ in
the graph, there exists $T_{\varepsilon} \in {\mathcal
  T_{\varepsilon}}$ such that $f_i(T_{\varepsilon}) \leq
(1+\varepsilon) f_i(T)$ for $i=1, \ldots, p$. Note that $[y_i \leq
y_i' \quad \forall i] \Rightarrow$ $\owa(y)$ $\leq$ $\owa$($y'$)
\cite{fodMR95}, and that $\owa$($(1+\varepsilon)y$) $= (1+\varepsilon)
\owa$($y$). Thus, a spanning tree $T^*_{\varepsilon}$ which is
OWA-optimal among the trees of ${\mathcal T_{\varepsilon}}$ satisfies
$\owa(f(T^*_{\varepsilon})) \leq (1+\varepsilon) \owa(f(T^*))$, where
$T^*$ is an $\owa$-optimal spanning tree in ${\mathcal T}$. We have
therefore an FPTAS for problem $\owa$-ST (the algorithm is indeed
polynomial since the search for $T^*_{\varepsilon}$ is performed in a
set of polynomial cardinality). The impact of this result is however
mainly theoretical, due to the high computational complexity of the
polynomial time procedure proposed by Papadimitriou and Yannakakis to compute an approximate Pareto set \cite{PapaY00}. For this reason,
we provide in the following more operational \emph{exact} algorithms (the
complexity of which is exponential in the worst case) for problem
$\owa$-ST.  The validity of the proposed algorithms depends on the
subclass of weights used in the $\owa$ criterion: strictly decreasing
or arbitrary (including the case of non-monotonic weights).

\section{Resolution method for strictly decreasing weights}

The strict decreasingness of the weights implies the convexity of the $\owa$ function. It is well-known that convexity of the objective function is a nice property in minimization problems. In particular, this property makes it possible to linearize objective function $\varphi = \owa$ in mathematical programming formulation $P$ of Section 2, so as to get a mixed integer programming formulation of problem $\owa$-ST. This is the topic of the following subsection.

\subsection{MIP formulation}
In the way of Yaman et al. \cite{YamKP01} for the robust spanning tree problem with interval data, we start from a compact mixed integer programming (MIP) formulation of the minimum spanning tree problem proposed by Magnanti and Wolsey \cite{MagnW95}. In this formulation the minimum spanning tree problem is considered as a special version of a flow problem. Every edge $e=[i,j]$ is replaced by two opposite arcs $(i,j)$ and $(j,i)$. The set of such arcs is denoted by $A$ in the sequel. Some vertex of the graph --- say 1 --- is selected as a source, and $n-1$ units of flow are incoming into it (where $V = \{1,\cdots,n\}$). Furthermore, 1 unit of flow is outgoing of every vertex $i \neq 1$. To each feasible flow in the directed graph, one can associate a connected partial graph by selecting edge $e=[i,j]$ as soon as there is at least one unit of flow on $(i,j)$ or $(j,i)$. By imposing that the number of selected edges is $n-1$, one obtains a spanning tree.
Let $\varphi_{ij}$ denote the flow on arc $(i,j)$, $x^e$ denote the boolean variable taking value 1 if $\varphi_{ij}>0$ or 
 $\varphi_{ji}>0$ (for $e =[i,j]$), and $v^e$ denote the scalar value of edge $e$. The corresponding MIP formulation of the minimum spanning tree problem is:

$$(P_{\mbox{\sc mst}})\left\{\begin{array}{ll}
 \min & \displaystyle \sum_{e\in E} v^e x^e\\
  \mbox{s.t.} & \displaystyle \sum_{(i,j) \in A} \varphi_{ij} - \displaystyle \sum_{(j,i) \in A} \varphi_{ji} = \left\{\begin{array}{ll}
	n - 1 & \mbox{if } i=1,\\
	-1 & \forall i \in V \setminus \{1\},
	\end{array}\right.\\ [5ex]	
&\varphi_{ij} \leq (n-1) x^e \quad \forall e =[i,j],	 \\
&\varphi_{ji} \leq (n-1) x^e \quad \forall e =[i,j],\\ [2ex]	
&\displaystyle \sum_{e \in E} x^e = 1,\\
&\varphi \geq 0, ~x^e \in \{0,1\} \quad \forall e \in E.	 
\end{array}\right.$$   

We modify and enrich this MIP formulation in order to take into account the multiobjective aspect of $\owa$-ST, as well as the $\owa$ aggregation function. The multiobjective aspect can be modelled by simply adding constraints $y_i = \sum_{e\in E} v_i^e x^e$ $\forall i \in \{1,\ldots,p\}$ into $P_{\mbox{\sc mst}}$, where $y_i$ is the variable corresponding to the value of the spanning tree on the $i^{th}$ component. The objective function is then $\owa(y) = \sum_{i=1}^p w_i y_{(i)}$, yielding a new program denoted by $P_{\mbox{\sc owa}}$. Nevertheless, in this formulation, this program is not linear. A nice linearization of the objective function has been proposed by Ogryczak and Sliwinski \cite{OgryS03}. The key idea is to use the Lorenz vector of $y$, the $i^{th}$ component of which is defined, in a minimization setting, by $L_i(y) = \sum_{j=1}^i y_{(j)}$. This notion has been introduced in economics for inequality comparisons (in a maximization setting), where $y$ is seen as an income distribution \cite[e.g.][]{Mouli91}. By noting that $\owa(y) = \sum_{i=1}^p (w_i - w_{i+1}) L_i(y)$, where $w_{p+1} = 0$, it appears that function $\owa$ is linear in the components of the Lorenz vector. Furthermore, component $L_i(y)$ is the solution of linear program $P_i^y$ defined below. However, when $y$ is a variable, program $P_i^y$ is no more linear. To overcome this difficulty, it is worth considering the dual version $D_i^y$, where $d_j^i$ (resp. $r_i$) are the dual variables for the inequality (resp. equality) constraints:
$$
(P_i^y)\begin{array}{ccc}
\left\{\begin{array}{ll}
\max & \sum_{j=1}^p \alpha_j^i y_j\\
\mbox{s.t.} & \sum_{j=1}^p \alpha_j^i = i,\\
& 0 \le \alpha_j^i \le 1 \quad \forall j \in \{1,\ldots,p\}.
\end{array}\right.
& ~~~~ & 
(D_i^y)\left\{\begin{array}{ll}
\min & i r_i + \sum_{j=1}^p d_j^i\\
\mbox{s.t.} & r_i + d_j^i \ge y_j \quad \forall j \in \{1,\ldots,p\},\\
& d_j^i \ge 0 \quad \forall j \in \{1,\ldots,p\}.
\end{array}\right.
\end{array}
$$ 
Replacing $L_i(y)$ by program $D_i^y$ for $i=1,\ldots,p$, program $P_{\mbox{\sc owa}}$ becomes:

$$(P_{\mbox{\sc owa}})\left\{\begin{array}{ll}
\min & \displaystyle\sum_{i=1}^p (w_i  - w_{i+1}) \displaystyle(i r_i + \sum_{j=1}^p d_j^i\displaystyle)\\
\mbox{s.t.} & r_i + d_j^i \ge y_j \quad \forall i,j \in \{1,\ldots,p\},\\
& y_i = \sum_{e \in E} v_i^e x^e$ $\forall i \in \{1,\ldots,p\},\\
 & \displaystyle \sum_{(i,j) \in A} \varphi_{ij} - \displaystyle \sum_{(j,i) \in A} \varphi_{ji} = \left\{\begin{array}{ll}
	n - 1 & \mbox{if } i=1,\\
	-1 & \forall i \in V \setminus \{1\},
	\end{array}\right.\\ [5ex]	
& \varphi_{ij} \leq (n-1) x^e \quad \forall e =[i,j],	 \\
& \varphi_{ji} \leq (n-1) x^e \quad \forall e =[i,j],\\ [2ex]	
& \displaystyle \sum_{e \in E} x^e = 1,\\
& d_i^j \ge 0 \quad \forall i,j \in \{1,\ldots,p\} \mbox{ and } r_i \mbox{ unrestricted},\\
&\varphi \geq 0, ~x^e \in \{0,1\} \quad \forall e \in E.	 
\end{array}\right.$$ 

Note that this formulation is valid only when $w_i-w_{i+1}>0$ for $i=1,\ldots,p$, which is the case here since the $w_i$'s are assumed to be strictly decreasing. With $m$ edges in the graph, the program involves $p^2+n+2m+1$ constraints and $p^2+p+3m$ variables (variables $y_i$'s can be omitted in the implementation): its size is therefore linear in the size of the input for a fixed number of objectives. In order to further reduce the number of variables and constraints involved in $P_{\mbox{\sc owa}}$, we now provide a preprocessing method that enables to reduce the density of the graph.

\subsection{Preprocessing of the graph}

As often done when presenting algorithms for constructing spanning trees
\cite[e.g.][]{Tarja83}, we describe our preprocessing phase as an edge
coloring process. Initially all edges are uncolored. We color one edge at a time either blue (mandatory) or red
(forbidden). At each step of the preprocessing phase, we have therefore a partial coloring $c(\cdot)$ of the edges of  
$G$. For a partial coloring $c$, we denote by $\Sol(c)$ the set of
trees including all blue edges, some of the uncolored ones (possibly
none), and none of the red ones. The aim of the preprocessing phase is to color as many edges as possible, so as to get a maximal coloring $c$ such that $\min_{T \in \Sol(c)} \owa(f(T)) = \min_{T \in \Sol} \owa(f(T))$. We now give conditions under which an edge can be colored blue or red, which are adaptations of the well-known \emph{cut optimality condition} and \emph{cycle optimality condition} to our problem. 

Before introducing the optimality conditions, we recall some definitions
from graph theory. A \emph{cut} in a graph is a partition of the vertices
into two disjoint sets and a \emph{crossing edge} (with respect to a cut)
is one edge that connects a vertex in one set to a vertex in the other one.
When there is no ambiguity, the term \emph{cut} will also be used to refer
to the set of crossing edges defined by the partition of the vertices.

\begin{prop}[optimality conditions]
Let $G$ be a connected graph with coloring $c$ of the edges. The following
properties hold:\\ %
$(i)$ \emph{Cut optimality condition.} Let us consider a cut $C$ in $G$ with
no blue edge. If there is some uncolored edge $e\in C$ such that $\owa(v^e - v^{e'}) \leq 0$ for all uncolored edges $e'\in C$, then $e$ belongs to an OWA-optimal tree in $\Sol(c)$.\\
$(ii)$ \emph{Cycle optimality condition.} Let us consider a cycle $C$ in
$G$ containing no red edge. If there is some uncolored edge $e\in C$ such
that $\owa(v^{e'}-v^e) \leq 0$ for all uncolored edges $e'\in C$, then $e$ can be colored red without changing value $\min_{T \in \Sol(c)} \owa(f(T))$.
\label{cond}
\end{prop}

\begin{proof}
The proof relies on the following property of function $\owa$ with non-increasing weights: $\owa(y-y') \ge \owa(y) - \owa(y')$. To prove this property, one uses the convexity of function $\owa$, which follows from the fact that $\owa(y) = \max_{\pi \in \Pi} \sum_i w_i y_{\pi(i)}$ where $\Pi$ is the set of all possible permutations of $(1,\ldots,p)$.
Combining convexity of $\owa$ and equality $\owa(\lambda y) = \lambda \owa(y)$, one deduces thus the required inequality: $\owa(y-y') + \owa(y') = 2 (\frac{1}{2}\owa(y-y')+\frac{1}{2}\owa(y')) \ge 2 \owa(\frac{1}{2}(y-y')+ \frac{1}{2}y') = \owa(y)$.

\emph{Proof of $(i)$.} Suppose there exists a cut $C$ and an uncolored crossing
edge $e$ that satisfies the cut optimality condition. Let $T$ be an $\owa$-optimal
spanning tree of $\Sol(c)$ that does not contain $e$. Now
consider the graph formed by adding $e$ to $T$. This graph has a cycle
that contains $e$, and this cycle must contain at least one other
uncolored crossing edge --- say $e'$ --- such that $e'\in C$, and
therefore we have $\owa(v^e-v^{e'}) \le 0$. We can get a new spanning tree $T' \in \Sol(c)$ by deleting $e'$ from $T$ and adding $e$. We claim that $\owa(f(T')) \le \owa(f(T))$. Cost $f(T')$ is indeed equal to $f(T) -
v^{e'} + v^e$. By the property indicated above, we have $\owa(f(T)-v^{e'}+v^e) -  \owa(f(T)) \le \owa(f(T)-v^{e'}+v^e - f(T)) = \owa(v^e-v^{e'})$. Consequently, since $\owa(v^e-v^{e'}) \le 0$, we deduce $\owa(f(T')) = \owa(f(T)-v^{e'}+v^e) \le  \owa(f(T))$, and $T'$ is therefore an $\owa$-optimal spanning tree in $\Sol(c)$ (that does contain $e$).

\emph{Proof of $(ii)$.} Suppose there exists a cycle $C$ containing
no red edge with an uncolored edge $e\in C$ such that $\owa(v^{e'}-v^e) \le 0$ for all
uncolored edges $e'\in C$. Let $T$ be an $\owa$-optimal
spanning tree of $\Sol(c)$ that contains $e$. Now consider the graph formed by removing $e$
from $T$. This graph is compounded of two connected components. The
induced cut contains at least one other uncolored crossing edge --- say
$e'$ --- such that $e' \in C$, and therefore $\owa(v^{e'}-v^e) \le 0$. We can get a new
spanning tree $T' \in \Sol(c)$ by deleting $e$ from $T$ and adding $e'$.
We claim that $\owa(f(T')) \le \owa(f(T))$. Cost $f(T')$ is indeed equal to $f(T) - v^e + v^{e'}$. By the property indicated above, we have $\owa(f(T) - v^e + v^{e'}) -  \owa(f(T)) \le \owa(v^{e'} - v^e)$. Consequently, since $\owa(v^{e'} - v^e) \le 0$, we deduce $\owa(f(T')) \le  \owa(f(T))$, and $T'$ is therefore an $\owa$-optimal spanning tree in $\Sol(c)$ (that does not contain $e$).
\end{proof}

Let us come back to the instance of Example~\ref{EXP1} to illustrate how to color edges according to these conditions.

\begin{example}
In the clique of Example~\ref{EXP1}, by considering cut $\{[1,4],[2,4],[3,4]\}]$, edge $[1,4]$ can be colored blue since $\owa(v^{[1,4]}-v^{[2,4]}) = \owa(-1,-4,1) = -0.6$ (with $w_1 = 0.5$, $w_2 = 0.3$ and $w_3 = 0.2$) and $\owa(v^{[1,4]}-v^{[3,4]}) = \owa(0,-3,1) = -0.1$. Furthermore, by considering cycle $\{[2,4],[2,3],[3,4]\}$, edge $[2,4]$ can be colored red since $\owa(v^{[2,3]}-v^{[2,4]}) = \owa(0,-2,0) = -0.4$ and $\owa(v^{[3,4]}-v^{[2,4]}) = \owa(-1,-1,0) = -0.5$. 
\end{example}

Note that Pareto-dominance of edge $e'$ over edge $e$ implies
$\owa(v^{e'}-v^{e}) \leq 0$ since all components of $v^{e'}-v^{e}$ are
negative in such a case (see edge $[2,4]$ in the previous
example). Our optimality conditions are thus an enrichment of the
multiobjective optimality conditions of Sourd and Spanjaard \cite{SourS08}, which are sufficient conditions for an edge to be mandatory or forbidden when generating the whole Pareto set. These conditions are indeed exclusively based on Pareto dominance between edges. We take here advantage of the knowledge of the OWA operator to optimize so as to be able to color a higher number of edges.

The single objective versions of these conditions make it possible to
design a generic greedy method \cite[e.g.][]{Tarja83}, from which
Kruskal's and Prim's algorithms can be derived. For problem $\owa$-ST, this method can be adapted so as to preprocess the graph prior to the resolution of the MIP formulation by a solver. The corresponding algorithm is indicated in Figure~\ref{onepass}. Obviously, and contrarily to the single objective case, this preprocessing method does not yield a spanning tree since binary relation $\owa(v^e-v^{e'}) \le 0$ does not induce a complete ordering of the edges in $E$. It enables however an important reduction of the number of variables (see numerical experiments), and therefore of the resolution time.

\begin{figure}
\begin{algo}{\sc Preprocessing($G$,$c$)}{A graph $G$ with coloring $c$ of edges}{A maximal coloring $c$}
\keyw{for each} edge $e$ of $E$ \keyw{do}\\
\tab \keyw{if} the cut optimality condition holds for $e$ \keyw{then}\\
\tab \tab set $c(e) = \blue$ \\
\tab \keyw{else} \keyw{if} the cycle optimality condition holds for $e$
\keyw{then}\\
\tab \tab set $c(e) = \red$ \\
\keyw{return} $c$
\end{algo}
\caption{\label{onepass} Preprocessing algorithm in the case of strictly decreasing weights.}
\end{figure}

The complexity of the preprocessing method strongly depends on the complexity of
detecting uncolored edges satisfying an optimality condition. In practice,
to determine whether an uncolored edge $e = [i,j]$ satisfies the cut
optimality condition, one performs a depth first search from $i$ in the
partial graph $G^\text{cut}_e = (V,E^\text{cut}_e)$ where $E^\text{cut}_e
= \{e' \in E\,|\,\owa(v^e-v^{e'})>0\} \cup \{e' \in E\,|\,c(e') =
\mbox{\blue}\}$. If $j$ does not belong to the set of visited vertices, then the
partition between visited and non-visited vertices constitutes a cut for
which $e$ satisfies the cut optimality condition. Similarly, to determine
whether an uncolored edge $e = [i,j]$ satisfies the cycle optimality
condition, one performs a depth first search from $i$ in the partial graph
$G^\text{cyc}_e = (V,E^\text{cyc}_e)$ where $E^\text{cyc}_e = \{e' \in
E\,|\,\owa(v^{e'}-v^e) \le 0\}\setminus \{e\} \cup \{e' \in E\,|\,c(e') =
\mbox{\blue}\}$. If $j$ is visited, then the chain from $i$ to $j$ in the
search tree, completed with $[i,j]$, constitutes a cycle for which $e$
satisfies the cycle optimality condition. Since the number of edges in the
graph is $m$ and the complexity of a depth first search is within $O(m)$
in a connected graph, the complexity of the overall preprocessing is $O(m^2)$.

\section{Resolution methods for arbitrary weights}
\label{sec:Bounds}

The linearization of the objective function is not valid anymore when
the weights are not strictly decreasing. Facing the difficulty induced
by the non-linearity of the objective function, we propose to resort
to a branch and bound procedure when the weights are not necessarily
strictly decreasing. For such procedures, the efficient computation of
bounds on the optimal value of problem $\owa$-ST is therefore worth
investigating. Besides, even when the weights are strictly decreasing, such bounds can be used in a shaving procedure (see Subsection 4.4 for a description) performed before the resolution phase in order to reduce the size of the instance. Thus, to obtain lower bounds, we present here two alternative relaxation methods for a multiobjective compromise search problem $P$: $\min
\varphi(y) \mbox{ s.t. } y=f(x),\, x \in X$, as defined in Subsection
\ref{sec:comp}.
\begin{itemize}
\item \emph{Relaxation of the image set:} it
  consists in defining a subspace $Y' \supseteq Y$ of the objective
  space (we recall that $Y = f(X)$) and solving problem $P_{Y'}$ defined by:      
  $$(P_{Y'}) ~~~\left\{
    \begin{array}{ll}
      \min ~~ \varphi(y)\\
      \mbox{s.t. } y  \in Y'
    \end{array}\right. $$    
\item \emph{Relaxation of the objective function:} it consists in
  defining a function $\varphi': \RP \rightarrow \R$ such that for all
  images $y$ in $Y$, $\varphi'(y) \leq \varphi(y)$ and solving problem $P_{\varphi'}$ defined by:
  $$(P_{\varphi'})~~~ \left\{
    \begin{array}{ll}
      \min ~~ \varphi'(y)\\
      \mbox{s.t. } y  =  f(x)\\
      ~~~~~x  \in  X
    \end{array}\right. $$ 
\end{itemize}     

  The main point is then to define $Y'$ (resp. $\varphi'$) such that
  the optimal value of problem $P_{Y'}$ (resp. $P_{\varphi'}$) can be
  efficiently computed and provide the tightest bound as possible. In
  this concern, it seems opportune to consider some continuous
  relaxations of $Y$ (resp. linear aggregation function
  $\varphi'$). In the following, we show more precisely how these relaxations can
  be performed when the aggregation function is an ordered weighted
  average. Note that, in both relaxations presented in the following,
  the computational efficiency of the procedure (to solve $P_{Y'}$ or $P_{\varphi'}$) only depends on the ability to quickly solve the single objective version of the
  problem. Consequently, these procedures can be applied to a broad
  spectrum of problems for which the single objective version can be
  solved efficiently. 

\subsection{Defining $Y'$ and solving problem $P_{Y'}$}

We now present a first bound on the value of an optimal solution to
$P$, obtained by relaxation of the image set. For this purpose,
let us define relaxed space $Y'$ in program $P_{Y'}$ as follows: 
$$
Y'=\{y \in \RP:y_i \geq b_i \quad \forall i=0,\ldots,p\}
$$
where $b_i = \min\{f_i(x):x \in X\}$ denotes the value of an optimal solution
according to $f_i$, and $f_0 = \sum_{i=1}^p f_i$. The values $b_i$'s
($i = 1, \ldots, p$) are obtained by $n$ runs of a standard algorithm
for the single objective version of the problem (provided it can be efficiently solved) for valuations $f_i$'s successively ($i = 1, \ldots,
p$). Concerning the computation of $b_0$, for simplicity purpose, we
only detail the case of multiobjective spanning tree problems. The
value $b_0$ is obtained by applying a standard minimal spanning tree
algorithm on the graph where each edge $e$ is valued by $v_0^e =
\sum_{i=1}^p v_i^e$: for any spanning tree characterized by $x$, we have indeed $f_0(x) = \sum_{i=1}^p f_i(x) =
\sum_{i=1}^p (\sum_{e \in E} v_i^e x^e) = \sum_{e \in E} (\sum_{i=1}^p
v_i^e) x^e = \sum_{e \in E} v_0^e x^e$. Note that this technique to compute
$b_0$ extends to many multiobjective versions of classical
combinatorial optimization problems (actually, as soon as the value of
a feasible solution is additively decomposable over its elements).

As indicated above, in order for the relaxation to be interesting in practice, one needs to provide an efficient procedure
to compute the optimal value of program $P_{Y'}$. When the objective function is $\owa$, one clearly needs to take into account the way
the components of solution $y$ are ordered for solving $P_{Y'}$. The OWA
function is actually piecewise linear. More precisely, it is linear
within each subspace of $\RP$ where all vectors are comonotonic,
i.e. for any pair $y,y'$ of vectors, there exists a permutation $\pi$
of $(1,\ldots,p)$ such that $y_{\pi(1)}$ $\geq$ $\ldots$ $\geq$
$y_{\pi(p)}$ and ${y'}_{\pi(1)}$ $\geq$ $\ldots$ $\geq$
${y'}_{\pi(p)}$. Problem $P_{Y'}$ can thus be divided into several
subproblems, each one focusing on a particular comonotonic subspace of
$\RP$. The solution of $P_{Y'}$ reduces to solving each linear program $P_{Y',\pi}$
defined by a particular permutation $\pi$ of $(1,\ldots,p)$:  
\vspace{-0.15cm}   
\begin{eqnarray*} (P_{Y',\pi}) \ \left \{ \begin{array}{ll}
\min\sum_{i=1}^p w_i y_{\pi(i)}  \\
\begin{array}{rrclllr}
\mbox{s.t.} & y_{\pi(i)} & \geq & y_{\pi(i+1)} & \forall i\in\{1,\ldots,p-1\}, &(1.1)\\
& y_i & \geq & b_i & \forall i\in\{1,\ldots,p\}, &(1.2)\\
& \sum_{i=1}^p y_i & \geq & b_0, & &(1.3)\\
& y & \in & \mathbb{R}^p.
\end{array}
\end{array} \right.
\end{eqnarray*} 
The value of the optimal solution $y^*_{Y'}$ to $P_{Y'}$ is then $\min_{\pi \in \Pi}
\owa(y^*_{Y',\pi})$ where $y^*_{Y',\pi}$ denotes the optimal
solution to linear program $P_{Y',\pi}$ and $\Pi$ the set of all
possible permutations. Note that for $p$ components, there are
$|\Pi|$=$p!$ linear programs to solve. However, in practice, it is not
necessary to solve the $p!$ linear programs. There exists indeed an
easily computable permutation $\pi^*$ for which $\owa(y^*_{Y'}) =
\owa(y^*_{Y',\pi^*})$: 

\begin{prop}[Galand and Spanjaard, 2007 \cite{GalaS07}]\label{PROP1}
  Let $\pi^*$ denote the permutation such that $b_{\pi^*(1)} \ge
  b_{\pi^*(2)} \ge \ldots \ge b_{\pi^*(p)}$. For any feasible
  solution $y$ to $P_{Y',\pi}$, there exists a feasible solution $y'$
  to $P_{Y',\pi^*}$ such that $\owa(y') = \owa(y)$.
\end{prop}

\begin{pf} The idea is to determine a feasible solution $y'$ to
  $P_{Y',\pi^*}$ such that $y'_{(i)}$=$y_{(i)}$ $\forall i$ (where $y_{(1)} \ge \ldots \ge y_{(p)}$). It
  implies indeed $\owa(y) = \owa(y')$ and the conclusion is then
  straightforward. In this respect, we construct a sequence
  $(y^j)_{j=1,\ldots,k}$ of solutions and a sequence
  $(\pi^j)_{j=1,\ldots,k}$ of permutations such that $y^j$ is feasible
  for $P_{Y',\pi^j}$ (for $j$=$1,\ldots,k$), with $y^1$=$y$,
  $\pi^1$=$\pi$, $\pi^k$=$\pi^*$ and $y^1_{(i)}$=$y^2_{(i)}$= $\ldots$
  =$y^k_{(i)}$ $\forall i$. Assume there exist $i_0,i_1\in$
  $\{1,\ldots,p\}$ such that $i_0$ $<$ $i_1$ and
  $b_{\pi^1(i_0)}$ $<$ $b_{\pi^1(i_1)}$. Let permutation $\pi^2$ be
  defined by $\pi^2(i_0)$=$\pi^1(i_1)$,
  $\pi^2(i_1)$=$\pi^1(i_0)$, and $\pi^2(i)$=$\pi^1(i)$ $\forall i$
  $\neq$ $i_0,i_1$. Let solution $y^2$ be defined by
  $y^2_{\pi^2(i)}$=$y^1_{\pi^1(i)}$ for $i$=$1,$ $\ldots,$ $p$. We now
  show that $y^2$ is a feasible solution to $P_{Y',\pi^2}$. Note first that
  $y^1_{\pi^1(i_0)}$ $\geq$ $b_{\pi^1(i_0)},$ $y^1_{\pi^1(i_1)}$
  $\geq$ $b_{\pi^1(i_1)},$ $y^1_{\pi^1(i_0)}$ $\geq$
  $y^1_{\pi^1(i_1)}$ and $b_{\pi^1(i_1)}$ $>$
  $b_{\pi^1(i_0)}$. Hence, constraints 1.2 are
  satisfied since:\\ 
  $\bullet$ $y^2_{\pi^2(i_0)}$=$y^1_{\pi^1(i_0)}$ $\geq$
  $y^1_{\pi^1(i_1)}$ $\geq$
  $b_{\pi^1(i_1)}$=$b_{\pi^2(i_0)}$,\\ 
  $\bullet$ $y^2_{\pi^2(i_1)}$=$y^1_{\pi^1(i_1)}$ $\geq$
  $b_{\pi^1(i_1)}$ $>$
  $b_{\pi^1(i_0)}$=$b_{\pi^2(i_1)}$,\\ 
  $\bullet$ $y^2_{\pi^2(i)}$=$y^1_{\pi^1(i)}$ $\geq$
  $b_{\pi^1(i)}$= $b_{\pi^2(i)}$ for $i\neq
  i_0,i_1$.\\ 
  Constraints 1.1 are also satisfied since $[y^1_{\pi^1(i)}$ $\ge$
  $y^1_{\pi^1(i+1)}\;$ $\forall i]$ $\Rightarrow$ $[y^2_{\pi^2(i)}$
  $\ge$ $y^2_{\pi^2(i+1)}$ $\forall i]$. Indeed, we have
  $y^1_{\pi^1(i)}$=$y^2_{\pi^2(i)}$ $\forall i$. These equalities
  imply also that constraint 1.3 is satisfied and that $y^2_{(i)}$=
  $y_{(i)}$ $\forall i$. Solution $y^2$ is therefore feasible for
  $P_{Y',\pi^2}$ with $y^2_{(i)}$=$y_{(i)}$ $\forall i$. Since any
  permutation is the product of elementary permutations, one
  always can construct in this way a sequence of permutations that
  leads to $\pi^*$ (and the corresponding feasible solutions). By
  setting $y'$=$y^k$, one obtains the desired feasible solution to
  $P_{Y',\pi^*}$. 
\end{pf} 

An immediate consequence of this result is that $\owa(y^*_{Y'}) =
\owa(y^*_{Y',\pi^*})$. Thus, the computation of $\owa(y^*_{Y'})$ 
reduces to solving linear program $P_{Y',\pi^*}$. For the sake of illustration, we present below an example.

\begin{example}
Let us come back to Example \ref{EXP1} where bounds $b_i$'s are defined from spanning trees $T_1$, $T_2$ and $T_3$ by $b_1=4$, $b_2=7$, $b_3=3$. The optimal spanning tree with respect to $f_0$ is spanning tree $T=\{[1,4],[2,3],[3,4]\}$ with $f_0(T) = 19$. The value of $b_0$ is then $19$.
It yields the following program $P_{Y'}$ (the values of $w_1$, $w_2$ and $w_3$ do not matter):
$$
\left\{
\begin{array}{ll}
\min & \mbox{ OWA}(y) = w_1 y_{(1)} + w_2 y_{(2)} + w_3 y_{(3)} \\
\mbox{s.t.} & y_1 \geq 4 \quad y_2 \geq 7 \quad y_3 \geq 3, \\
& y_1 + y_2 + y_3 \geq 19, \\
& y_1 \in {\mathbb R},\,y_2 \in {\mathbb R},\,y_3 \in {\mathbb R}.
\end{array}
\right.
$$
As a consequence of Proposition~\ref{PROP1}, solving $P_{Y'}$ amounts to solving linear program $P_{Y',\pi^*}$ defined by:
$$
\left\{
\begin{array}{ll}
\min & w_1 y_2 + w_2 y_1 + w_3 y_3 \\
\mbox{s.t.} & y_{2} \geq y_{1} \geq y_{3}, \\
& y_1 \geq 4 \quad y_2 \geq 7 \quad y_3 \geq 3, \\
& y_1 + y_2 + y_3 \geq 19, \\
& y_1 \in {\mathbb R},\,y_2 \in {\mathbb R},\,y_3 \in {\mathbb R}.
\end{array}
\right.
$$
We have indeed $\pi^*(1) = 2$, $\pi^*(2) = 1$ and $\pi^*(3) = 3$ since $b_2 \geq b_1 \geq b_3$. 
When $w_1 = 0.5$, $w_2=0.3$ and $w_3=0.2$, the value of the bound provided by programm $P_{Y',\pi^*}$ is therefore $6.5$ (we recall that the value of the $\owa$-optimal spanning tree is $7$).
\end{example}

\subsection{Defining $\varphi'$ and solving problem $P_{\varphi'}$}
\label{PPHI'}

The second bound for $P$ is obtained by relaxation of the objective function. The idea is to use a linear function $\varphi'$ defined by $\varphi'(y)=\sum_{i=1}^p \lambda_i y_i$ where the vector $\lambda=(\lambda_1,\ldots,\lambda_p)$ of coefficients satisfy the following constraints:
\begin{equation} \label{EQ1}
\sum_{i \in I} \lambda_i \leq \sum_{i=1}^{\vert I \vert} w_i \quad \mbox{ for all subset } I \subseteq \{1,\ldots,p\}  \mbox{ of objectives}
\end{equation}
When using such weight vectors, we have indeed $\varphi'(y) \leq \varphi(y)$ for all $y$ in $Y$, as established by the following: 

\begin{prop}\label{PROP2}
Let $\Lambda$ denote the set of all vectors $\lambda \in \RP$ satisfying Constraint~\ref{EQ1}. If $\lambda \in \Lambda$, then for all $y \in \RP$, $\sum_{i=1}^p \lambda_i y_i \leq \owa(y)$.
\end{prop}

\begin{pf}
Let $\pi$ denote a permutation such that $y_{\pi(1)} \ge \ldots \ge y_{\pi(p)}$. Note first that $\owa(y) = \sum_{i=1}^p w_i y_{\pi(i)} = \sum_{i=1}^p (\sum_{j=1}^i w_j) (y_{\pi(i)} - y_{\pi(i+1)})$ where $y_{\pi(p+1)} = 0$. By Constraint~\ref{EQ1} with subset of objectives $I = \{\pi(1),\ldots,\pi(i)\}$, we have $\sum_{j=1}^i w_j \geq \sum_{j=1}^i \lambda_{\pi(j)}$ $\forall i$. Thus $\sum_{i=1}^p (\sum_{j=1}^i w_j) (y_{\pi(i)} - y_{\pi(i+1)}) \geq \sum_{i=1}^p (\sum_{j=1}^i \lambda_{\pi(j)}) (y_{\pi(i)} - y_{\pi(i+1)})$. Since $\sum_{i=1}^p (\sum_{j=1}^i \lambda_{\pi(j)}) (y_{\pi(i)} - y_{\pi(i+1)}) = \sum_{i=1}^p \lambda_{\pi(i)} y_{\pi(i)}$ and $\sum_{i=1}^p \lambda_{\pi(i)} y_{\pi(i)} = \sum_{i=1}^p \lambda_i y_i$, we have $\owa(y) \geq \sum_{i=1}^p \lambda_i y_i$.
\end{pf}

This result can be seen as a specific instanciation of a more general result of Schmeidler \cite{Schme86,Schme89} on Choquet integrals  in decision under uncertainty (OWA being a particular subclass of Choquet integrals). Schmeidler's result has been used to provide bounds in Choquet-based optimization under uncertainty \cite{GalaP07}, as well as under multiple objectives \cite{GalPS09}. Note that, in the case of decreasing weights in $\owa$, a useful by-product of Schmeidler's result is the existence of a normalized set of weights (i.e., summing up to 1) such that Proposition~\ref{PROP2} holds: for instance, when setting $\lambda_i = 1/p$, we have $\sum_i (1/p)y_i \leq \owa(y)$ by Chebyshev's sum inequality. Nevertheless, when the weights are not decreasing, the existence of normalized weights satisfying Constraint~\ref{EQ1} is not guaranteed.

Solving problem $P_{\varphi'}$ can be efficiently done by valuing every edge $e$ by $\sum_{i=1}^p \lambda_i v^e_i$ and then performing a standard minimum spanning tree algorithm. Besides, in order to obtain the best lower bound as possible, the optimal set of weights $\lambda \in \Lambda$ (i.e., providing the best lower bound according to Proposition~\ref{PROP2}) can be obtained by solving the following program:
\vspace*{-0.2cm}
\begin{eqnarray*}
 \max_{\lambda \in \mathbb{R}^p} & z(\lambda) &  = \min_{y \in Y} \sum_{i=1}^p \lambda_i y_i, \\
 & \mbox{s.t. } & \sum_{i \in I} \lambda_i \leq \sum_{i=1}^{\vert I \vert} w_i \quad \forall I \subseteq \{1,\ldots,p\},\\
& & \lambda_i \geq 0 \quad \forall i = 1, \ldots, p. \label{const2}
\end{eqnarray*}

Given that $z$ is a concave piecewise linear function of $\lambda$ for a fixed $y$ (since it is the
lower envelope of a set of linear functions $\{\sum_{i=1}^p y_i
\lambda_i : y \in Y\}$), we solve this program by using the SolvOpt
library \cite{KappK00}, which is an implementation of Shor's
$r$-algorithm \cite{Shor85}. This algorithm is indeed especially convenient for
non-differentiable optimization, and the implemented SolvOpt library
enables to perform constrained optimization. This approach is closed
to the lower bounding procedure proposed by Punnen and Aneja \cite{PunnA95} for min-max combinatorial optimization. In broad outline, it can be viewed as a sequence of minimum spanning tree computations according to a varying weight vector $\lambda$, until convergence towards a weight vector maximizing $\min_{y \in Y} \sum_{i=1}^p \lambda_i y_i$.

\begin{example}
Let us come back to the instance of Example \ref{EXP1}. Running Shor's $r$-algorithm yields $\lambda = (0.28,0.5,0.22)$ and $y = (6,8,6)$ (the image of spanning tree $\{[1,2],[1,4],[2,3]\}$, which is actually OWA-optimal). The bound is therefore $0.28 \times 6 + 0.5 \times 8 + 0.22 \times 6 = 7$ (which is the value of an OWA-optimal tree). This is better than the bound obtained by the previous relaxation (value 6.5). The computational burden is however more important, since it requires to solve much more single objective problems.
\end{example}

\subsection{Branch and bound procedure}

The two bounds presented above can of course be inserted into a branch and bound procedure in order to determine an OWA-optimal spanning tree for arbitrary weights. We summarize below the main features of the branch and bound procedure we propose.\\ [1ex]
\emph{Branching scheme.} The branching scheme is very simple: at each node, an edge $e$ of $G$ is selected and two subproblems are created. In the first one, edge $e$ is mandatory while in the second one, edge $e$ is forbidden. The heuristic to select the next edge consists in searching for an edge $e$ such that $\sum_{i=1}^p v^e_i$ is minimal among the remaining ones.\\ [1ex]
\emph{Computing the lower bound.} The lower bound at each node of the branching tree is computed by relaxation of the image set (by generating and solving problem $P_{Y'}$) or by relaxation of the objective function (by solving a sequence of problems $P_{\varphi'}$).\\ [1ex] 
\emph{Updating the incumbent.} When defining and solving problem
$P_{Y'}$ or $P_{\varphi'}$, the algorithm checks whether a newly
computed spanning tree is better than the current best known spanning
tree according to $\owa$. We take indeed advantage of the property that feasible spanning trees are generated when computing the lower bound: either $p+1$ spanning trees are generated to obtain the values of $b_i$ ($i=0,\ldots,p$) in $P_{Y'}$
or a sequence of spanning trees is generated during the running of
Shor's $r$-algorithm to find the best possible bound according to
$P_{\varphi'}$.\\ [1ex] 
\emph{Initialization.} A branch and bound algorithm is notoriously more efficient when a good solution is known even before starting the search. The initial solution in our method is the best known solution (according to $\owa$) after the run of the shaving procedure described below.

\subsection{Shaving procedure}
The definition of lower bounds makes it possible to resort to a \emph{shaving} procedure in order to reduce the size of the instance before running the main algorithm. The term ``shaving'' was introduced by Martin and Shmoys \cite{MartinShmoys96} for the job-shop scheduling problem. Assuming at least a feasible solution is known, this procedure works as follows: for each edge $e$, we build a subproblem in which $e$ is made mandatory. If the computation of the lower bound proves that the subproblem cannot improve the current best known solution, then it means that $e$ can be made definitively forbidden (colored red). Conversely when $e$ has not been colored red by the previous procedure, we test similarly whether $e$ can be made definitively mandatory (colored blue). Note that, here again, when computing the lower bounds, one checks whether a newly detected spanning tree improves the current best known solution (according to $\owa$). Of course, the shaving procedure is all the more efficient as a good feasible solution is initially known. For this purpose, we generate $k$ feasible solutions by running a $k$-best ranking algorithm (i.e., returning the $k$ best solutions) for the minimum spanning tree problem on the instance valued by the arithmetic mean of the vectors. The choice of $k$ depends on the size of the instance.

\begin{example}
Let us come back again to the clique of Example~\ref{EXP1}. Assuming that
the $k$-best ranking algorithm is run for $k=2$, the current best known solution is then spanning tree $T_4 = \{[1,2],[1,4],[2,3]\}$. Let us now
simulate the progress of the shaving procedure. Making edge $[1,2]$
forbidden, the minimum spanning tree for weight vector $\lambda =
(0.3,0.5,0.2)$ is $\{[1,3],[1,4],[2,3]\}$ with image $y =
(7,9,4)$. The induced lower bound is 7.4, which is strictly greater
than the OWA-value of $T_4$ ($\owa(f(T_4))=7$). Therefore edge $[1,2]$
is colored blue. Then making edge $[1,3]$ mandatory, the minimum spanning
tree for weight vector $\lambda = (0.5,0.3,0.2)$ is
$\{[1,2],[1,3],[1,4]\}$ with image $y = (8,7,6)$. The induced lower
bound is 7.3, which is strictly greater than the OWA-value of
$T_4$. Therefore edge $[1,3]$ is colored red. 
Next, making edge $[1,4]$ forbidden, the minimum spanning tree for
weight vector $\lambda = (0.3,0.5,0.2)$ is $\{[1,2],[2,3],[3,4]\}$
with image $y = (6,11,5)$. The induced lower bound is 8.3, and edge
$[1,4]$ is therefore colored blue. Finally, making edge $[2,3]$
forbidden, the minimum spanning tree for weight vector $\lambda =
(0.2,0.5,0.3)$ is $\{[1,2],[1,4],[3,4]\}$ with image $y =
(5,9,6)$. The induced lower bound is 7.3, and consequently edge
$[2,3]$ is colored blue. All edges of $T_4$ are colored blue: one concludes
that $T_4$ is an $\owa$-optimal spanning tree without even starting the
main algorithm. 
\end{example}

\section{Experimental results} \label{STests}

Before we study more carefully the behavior of our algorithms, we give some insights into previous results on related topics. The most widely studied related topic is the generation of the Pareto set in the bi-objective spanning tree problem \cite{AndJL96,HamaR94,RaASG98,SourS08,SteiR08}. To our knowledge, the most efficient algorithm for this problem enables to solve randomly drawn clique instances containing up to 400 vertices \cite{SourS08} (note that these results significantly improved the size of the instances that could be handled, since it grew from 40 to 400 vertices). However these results do not extend to more than two objectives. More generally, even when looking for a single compromise solution within the Pareto set, it seems that there is very few available numerical experiments in the literature for more than two objectives. Although several works deal with the min-max spanning tree problem (i.e. determining a spanning tree minimizing the max criterion) \cite{AisBV07,AisBV09,HamaR94,Warbu85,Yu98}, the content of these works is indeed mainly theoretical. Actually, the only numerical results we know for more than two objectives are devoted to the determination of a Choquet-optimal spanning tree \cite{GalPS09}. The size of the tackled instances goes from 30 to 70 vertices according to the number of objectives and the parameterization of the Choquet integral.

\subsection{Experimental details} \label{SSImpl}

All the algorithms have been implemented in C++ and were run on a 2.6
GHz personal computer with a memory of 3.4GB. The test instances are defined as follows. All considered graphs are cliques. The components of cost vectors are randomly drawn between 1 and 100 on each edge. The number of objectives varies from 3
to 10, and the number of vertices from 10 to 100 for strictly decreasing weights, and from 10 to 25 for non-monotonic weights. For each kind of
instances (depending on the number of nodes and on the number of
objectives), 30 instances were randomly drawn to obtain average results. 

\subsection{Tests with strictly decreasing weights}

The global procedure to solve problem $P_{\mbox{\sc owa}}$ when the
weights are strictly decreasing consists of two phases: 
\begin{enumerate}
\item \emph{Coloration phase}: making the most possible edges blue (mandatory) or
  red (forbidden) by running first the preprocessing procedure, and then a shaving procedure taking into account the coloration obtained by preprocessing. 
\item \emph{Resolution phase}: determining the $\owa$-optimal spanning
  tree by running the main resolution algorithm (solution by MIP or
  branch and bound) on the reduced instance.
\end{enumerate}

In this section, one summarizes the results one has obtained by
running this global procedure. For initializing the shaving procedure in the coloration phase, the $k$-best ranking algorithm proposed by Katoh et al. \cite{KatIM81} is launched for $k$ varying from 500 to 5000 depending on the size of the instance. In all cases, the order of magnitude of its running time is about a few milliseconds. Concerning the shaving itself, preliminary
results have shown that it is much more efficient
when the bound is defined by relaxation of the objective function than by the
relaxation of the image set, especially when the
size of the instances grows. For this reason, one only summarizes here
the execution times obtained when the shaving procedure is performed
by relaxation of the objective function. This shaving procedure is
denoted by \emph{sh$_{\varphi'}$} in the sequel. Finally, one compares the various
resolution algorithms proposed in this paper: the solution by MIP and the one by
branch and bound (two versions according to the bound adopted). The mixed integer program is solved by using solver ILOG CPLEX 11. Regarding the branch and bound procedure, the bounds proposed in Section 4 are compared: the one obtained by relaxation of the objective function, and the one obtained by relaxation of the image set. 

Tables \ref{tabPL1}, \ref{tabPL2} and \ref{tabPL3} summarize the results
obtained by running the algorithms on cliques with respectively 3, 5
or 10 objectives (the weights of the $\owa$ operator are indicated in the caption of the tables), for various numbers of vertices for which the average resolution time is lower than 30 minutes. The upper part of the tables summarizes the informations
about the coloration phase (line \emph{pp} for the preprocessing
procedure and line \emph{sh$_{\varphi'}$} for the shaving procedure). For each procedure and each size of instance, the average execution time is indicated. Furthermore, below line \emph{pp} (resp. \emph{sh$_{\varphi'}$}), the average number of edges made blue ($\#b$) or red ($\#r$) after preprocessing (resp. after preprocessing \emph{and} shaving) is indicated (couple $(\#b-\#r)$). In the lower part of the tables are indicated the average total resolution times for various combinations of procedures for the coloration phase and the resolution phase. To evaluate the variability of the resolution time, the minimal running time ($\min$) as well as the maximal one ($\max$) are also indicated (couple $\min-\max$ under the average execution time). To encode the combinations, the following abbreviations are used: MIP stands of course for Mixed Integer Programming, $BB_{Y'}$  for the branch and bound obtained by relaxation of the image set, and $BB_{\varphi'}$ for the branch and bound obtained by relaxation of the objective function. Note that, in order to show the impact of the coloration phase, the first line of the lower part indicates the resolution times when solving the MIP formulation without resorting to any preprocessing or shaving. 

\begin{table}[!h]
  \begin{center}
{\small
  \begin{tabular}{|c|ccccc|}    \hline
      \rule[1pt]{0pt}{13pt} 
   $n$ & 20&30&40&50&60\\
    \hline
    pp  &0&0.02&0.05&0.11&0.21\\
    	 &(2.8 - 128.3)&(4.4 - 332)&(5.8 - 628.3)&(6.3 - 1025.4)&(7.1 - 1524.9)\\
    sh$_{\varphi'}$&0.51&2.89&10.76&31.43&79.19\\
         &(12 - 160.7)&(20.1 - 390)&(26.1 - 720.5)&(33.5 - 1154.7)&(37.9 - 1680.4)\\
   \hline
 MIP&0.51&2.8&8.11&56.11&1018.5\\
      &0.02 - 4.65&0.2 - 20.3&0.45 - 49.7&0.88 - 639.9&1.58
      - 16842\\  
pp+sh$_{\varphi'}$+MIP&0.55&3.04&11.73&32.26&89.38\\
   &0.31 - 0.99&2.26 - 4.93&7.88 - 32.55&25.7 - 43.4&68.8
   - 177.3\\ 
pp+sh$_{\varphi'}$+BB$_{Y'}$&0.58&3.11&25.89&34.51&241.52\\
     &0.31 - 1.49&2.26 - 5.79&7.87 - 345.4&25.7 - 67.7&68.8 - 1945.5\\
pp+sh$_{\varphi'}$+BB$_{\varphi'}$&1.35&10.27&72.1&206&1049.7\\
    &0.32 - 3.79&2.71 - 37.2&8 - 571.2&28.5 - 1183.3&75.8 - 4457.9\\
  \hline
    \end{tabular}\\ [2ex]
\begin{tabular}{|c|cccc|}   \hline
 \rule[1pt]{0pt}{13pt} 
  $n$ & 70&80&90&100\\
  \hline
  pp & 0.34&0.54&0.79&1.14\\
    &(8.5 - 2125.2)&(9.3 - 2827.8)&(12 - 3624.1)&(13.2 - 4520.8)\\
  sh$_{\varphi'}$&165.6&314.7&563.7&987.7\\
  &(50.1 - 2318)&(52.8 - 3044.7)&(63.1 - 3879.1)&(71.6 - 4811.4)\\
  \hline
 pp+sh$_{\varphi'}$+MIP&266.94&342.79&590.98&1167.83\\
 &138.3 - 3206.5&267.1 - 720.6&475.2 - 1025.3&810.6 - 4453\\
\hline
\end{tabular}
}
\end{center}
\caption{\label{tabPL1} Synthesis of the numerical results for 3 objectives ($w_1 = 0.6$, $w_2 = 0.3$ and $w_3 = 0.1$). Execution times are indicated in seconds.
}
\end{table}

Table 1 shows that the best results are obtained by using the MIP formulation with a coloration phase. This is the only algorithm able to handle instances with more than 60 vertices with an average computation time below 30 minutes. It is essential to note that the coloration phase has a very significative impact on the computation time. For $n= 60$, one sees indeed that the resolution time for MIP is above 1000 seconds, while it is below 100 seconds when the coloration phase is used. This can be easily understood by observing that the number of colored edges at the end of the coloration phase is a low fraction of the initial number of edges. For instance, in average, it remains only 67 uncolored edges over 4950 for $n=100$. The main interest of preprocessing is that it enables to color a lot of edges in a very low computation time. Hence, the number of edges that are tested during the shaving procedure is reduced, which is computationally interesting since the shaving is more powerful but takes also much more time. Performing shaving after preprocessing makes it possible to color blue a large part of the optimal spanning tree in a reasonable computation time.  For illustration, in average,  preprocessing colors red 4520.8 edges for $n=100$, and shaving colors blue 71.6 edges (including those colored blue by preprocessing) over a maximum number of 99. A large amount of the total resolution time is spent in the coloration phase: for $n=100$, 988.84 seconds are spent in average in the coloration phase (pp+sh$_{\varphi'}$) over a total resolution time of 1167.83 seconds. Finally, note that there is an important variability in the resolution time according to the instance (for $n=60$, the MIP algorithm takes between 1.58 seconds and more than 4 hours and half). The coloration phase tends however to reduce this variability.

\begin{table}[!h]
  \begin{center}
 {\small \begin{tabular}{|c|cccccc|}
   \hline
      \rule[1pt]{0pt}{13pt} 
   $n$ & 10&20&30&40&50&60 \\
    \hline
 pp  &0&0.01&0.02&0.065&0.14&0.27\\
    	  &(0.16 - 8.36)&(0.18 - 63.42)&(0.27 - 187.2)&(0.16 - 384.9)&(0.27 - 668)&(0.3 - 1034.9)\\
 sh$_{\varphi'}$& 0.13&1.71&10.5&43.57&132.21&337.13\\
    	  &(3.54 - 26.84)&(5.42 - 142.5)&(8.13 - 359.3)&(9.92 - 667.1)&(12.77 - 1079.7)&(15.52 - 1596.3)\\
    \hline
 MIP&0.07&1.74&52.39&133.8&1426.7&-\\
    &0.03 - 0.14&0.12 - 12.06&0.59 - 1166.2&0.94 - 3591&12.26 - 16610& - \\
  pp+sh$_{\varphi'}$+MIP &0.16&2.63&27.27&87.66&433.8&1202.7\\
      & 0.09 - 0.24&1.47 - 12.83&8.66 - 284.2&38.88 - 693&110.5 - 2589.6&301.8 - 10153.4\\
  pp+sh$_{\varphi'}$+BB$_{Y'}$& 0.17&14.23&764&-&-&-\\
 & 0.09 - 0.62&1.5 - 218.8&8.75 - 13767.5& - & - & - \\
pp+sh$_{\varphi'}$+BB$_{\varphi'}$& 0.37&19.63&541.3&-&-&-\\
    &0.09 - 2.02&2.58 - 91.58&13.52 - 5018.5& - & - & - \\
  \hline
       \end{tabular}
}
\end{center}
\caption{\label{tabPL2} Synthesis of the numerical results for 5 objectives ($w_1 = 0.5$, $w_2 = 0.3$,
  $w_3 = 0.1$, $w_4 = 0.06$ and $w_5 = 0.04$). Execution times are indicated in seconds. Symbol ``-'' means that the average execution time exceeds 30 minutes.}
\end{table}

\begin{table}[!h]
  \begin{center}
 {\small  \begin{tabular}{|c|ccccc|}
    \hline
      \rule[1pt]{0pt}{13pt} 
   $n$ & 10&15&20&30&35\\
    \hline
 pp  & 0&0&0.01&0.03&0.05\\
    	 &(0 - 2.22)&(0 - 8.73)&(0.03 - 20.2)&(0.03 - 69.5)&(0.03 - 108.5)\\
 sh$_{\varphi'}$& 2.69&5.21&11.41&44&86.57\\
    	 &(3.78 - 25.58)&(2.43 - 58.33)&(1.83 - 107.07)&(1.1 - 256.7)&(1.3 - 365.27)\\
    \hline
 MIP&0.15&1.87&11.47&519.2&-\\
    & 0.05 - 0.55&0.13 - 7.65&0.6 - 33.65&4.76 - 6623& - \\
  pp+sh$_{\varphi'}$+MIP& 2.77&6.54&19.29&311&1284.8\\
  &  2.32 - 3.58&4.42 - 12.52&10.2 - 40.6&41.26 - 2937.3&188.1 - 8504.6\\
  pp+sh$_{\varphi'}$+BB$_{Y'}$ & 3.15&91.74&-&-&-\\
  & 2.32 - 8.47&4.41 - 532.71& - & - & - \\
pp+sh$_{\varphi'}$+BB$_{\varphi'}$& 7.93&58.27&477.4&-&-\\
    & 2.32 - 25.75&5.21 - 317&13.32 - 2109.5& - & - \\
  \hline
    \end{tabular}
}
\end{center}
\caption{\label{tabPL3} Synthesis of the numerical results for 10 objectives ($w_1 = 0.25$, $w_2 = 0.2$,
  $w_3 = 0.15$, $w_4 = 0.1$, $w_5 = 0.09$, $w_6 = 0.08$, $w_7 = 0.06$,
  $w_8 = 0.04$, $w_9 = 0.02$ and $w_{10} = 0.01$). Execution times are indicated in seconds. Symbol ``-'' means that the average execution time exceeds 30 minutes.}
\end{table}

Tables \ref{tabPL2} and \ref{tabPL3} show that the number of objectives has a great impact on the performances of the algorithms. For instance, for $n = 60$, the resolution time of algorithm pp+sh$_{\varphi'}$+MIP is 1202.7 seconds in average for 5 objectives, compared to 89.38 seconds for 3 objectives. Besides, for 10 objectives, one can only handle instances with up to 35 vertices within the time limit of 30 minutes.
The resolution time is also sensitive to the weights of the $\owa$ operator. In order to evaluate this sensitivity, we also performed a test with weights inducing an $\owa$
operator closer to the arithmetic mean. Namely, we set $w_1 = 0.4$, $w_2 = 0.35$ and $w_3 = 0.25$.

\begin{table}[!h]
  \begin{center}
  {\small \begin{tabular}{|c|ccccc|}    \hline
      \rule[1pt]{0pt}{13pt} 
   $n$ & 40&50&60&70&80\\
    \hline
    pp  &0.06&0.12&0.22&0.38&0.59\\
    	 &(23.57 - 717.6)&(30.47 - 1148.5)&(35.9 - 1675.2)&(41.85 - 2305.7)&(48 - 3033.6)\\
    sh$_{\varphi'}$&2.75&7.22&18.65&37.51&74.3\\
         &(32.5 - 733.9)&(41.9 - 1167.3)&(49.3 - 1700.2)&(57.1 - 2334.9)&(64.85 - 3064.8)\\
   \hline
 MIP&7&65.24&243.8&-&-\\
      &0.25 - 94.75&0.57 - 830.3 &7.12 - 3227& - & - \\
pp+sh$_{\varphi'}$+MIP&2.81&7.35&18.9&37.96&75.21\\
  &2.04 - 3.96&4.22 - 11.37&14.23 - 22.75&26.26 - 48.95&55 - 95.15\\
pp+sh$_{\varphi'}$+BB$_{Y'}$&2.82&7.38&18.97&38.376&77.85\\
     &2.04 - 4.01&4.22 - 11.86&14.23 - 23.04&26.26 - 55.59&54.99 - 104.1\\
pp+sh$_{\varphi'}$+BB$_{\varphi'}$&5.37&13.28&51.5&114.9&391.1\\
    &2.1 - 17.03&4.33 - 48.46&15.57 - 184.5&26.57 - 484.1&64.08 - 1853.1\\
  \hline
    \end{tabular}\\ [2ex]
}
\end{center}
\caption{\label{tabPL26} Synthesis of the numerical results for 3 objectives ($w_1 = 0.4$, $w_2 = 0.35$ and $w_3 = 0.25$). Execution times are indicated in seconds. 
}
\end{table}

For this set of weights, all the tested algorithms perform better, as it can be seen in
Table~\ref{tabPL26}. In particular, the preprocessing procedure makes it possible to color
blue more edges than for the previous set of weights. For example, for $n = 60$, 49.3 edges are colored blue
in average (over a maximal number of blue edges of 59)
while only 37.9 are colored blue in average when the
previous set of weights is used. Moreover, one can also observed that,
once the coloration phase is performed, branch and bound procedure
BB$_{Y'}$ is as efficient as the
resolution by MIP.

\subsection{Tests with non-monotonic weights}

We present here the results obtained when the weights are
non-monotonic. In this subsection, we used Hurwicz's criterion as
$\owa$ operator. We recall that, for a vector $y$, it is defined as
$\alpha \max_i y_i + (1-\alpha) \min_i y_i$. The instances are defined
similarly to the previous subsection, with 3, 4 or 5 objectives, and
parameter $\alpha$ varying from 0.4 to 0.6 in Hurwicz's criterion. As
mentioned in Section 3, the MIP formulation is no more valid in this
case, as well as the preprocessing procedure (the optimality
conditions do not hold anymore). Furthermore, the bound obtained by
relaxation of the objective function becomes weak since there does not
necessarily exist a normalized set of weights satisfying
Constraint~\ref{EQ1} (see Section~\ref{PPHI'}). Consequently, in the
global procedure described in Subsection 5.2, the coloration phase consists of a
single run of the shaving procedure with the bound obtained by
relaxation of the image set, and the resolution phase consists of
applying branch and bound BB$_{Y'}$. Tables~\ref{tabPL4}, \ref{tabPL5}
and \ref{tabPL6} summarize the results obtained. The conventions are the same than in the previous subsection.

\begin{table}[!h]
  \begin{center}
 {\small 
 \begin{tabular}{|c|c|ccccc|}
    \hline
      \rule[1pt]{0pt}{13pt} 
       &$n$ & 5&10&15&20&25\\
       \hline
  \hline
       & sh$_{Y'}$&0&0&0.02&0.04&-\\
       $\alpha = 0.4$   & &(1.5 - 3.23)&(0.13 - 12.53)&(0 - 4)&(0 - 2.73)& - \\
       \cline{2-7} 
       &sh$_{Y'}$+BB$_{Y'}$&0&0.05&1.01&37.93&-\\
       &  & 0 - 0.02&0 - 0.77&0.05 - 6.18&0.85 - 129.7& - \\
  \hline
  \hline
   & sh$_{Y'}$&0&0&0.02&0.04&0.07\\
$\alpha = 0.5$  & &(1.72 - 3.22)&(0.07 - 10.83)&(0.03 - 18.5)&(0 - 15.27)&(0 - 12.57)\\

    \cline{2-7} 
 &sh$_{Y'}$+BB$_{Y'}$&0&0.03&0.87&7.53&195.6\\
  &  &0 - 0.03&0 - 0.14&0.06 - 9.93&0.23 - 40.46&2.14 - 2203.1\\
  \hline
  \hline
  & sh$_{Y'}$&0&0.01&0.01&0.03&0.06\\
  $\alpha = 0.6$  & &(1.73 - 3.3)&(1.9 - 23.2)&(0.8 - 48.07)&(0.23 - 82.1)&(0 - 100.5)\\
  \cline{2-7} 
  &sh$_{Y'}$+BB$_{Y'}$&0&0.01&0.11&0.66&5.09\\
  &  &0 - 0.01&0 - 0.04&0.01 - 0.66&0.07 - 3.02&0.19 - 23.68\\
  \hline
    \end{tabular}
}
\end{center}
\caption{\label{tabPL4} Synthesis of the numerical results for 3 objectives. Execution times are indicated in seconds. Symbol ``-'' means that the average execution time exceeds 30 minutes.}
\end{table}

\begin{table}[!h]
  \begin{center}
  {\small \begin{tabular}{|c|c|ccccc|}
    \hline
       \rule[1pt]{0pt}{13pt} 
       &$n$ & 5&10&15&20&25\\
       \hline
     \hline
    & sh$_{Y'}$&0&0.01&0.02&-&-\\
       $\alpha = 0.4$   & &(0.9 - 2.6)&(0 - 2.67)&(0 - 4.33)& - & - \\

       \cline{2-7} 
       &sh$_{Y'}$+BB$_{Y'}$&0&0.21&10.31&-&-\\
       &  &0 - 0.02&0.02 - 0.57&0.21 - 82.37& - & - \\
  \hline
  \hline
   & sh$_{Y'}$&0&0.01&0.02&0.05&-\\
$\alpha = 0.5$  & &(0.63 - 2.03)&(0 - 4.53)&(0 - 4.9)&(0 - 0.47)& - \\

    \cline{2-7} 
 &sh$_{Y'}$+BB$_{Y'}$&0&0.12&6.03&669.09&-\\
  &  &0 - 0.02&0.01 - 0.72&0.22 - 63.63&10.84 - 8779.1& - \\
  \hline
  \hline
  & sh$_{Y'}$&0&0.01&0.02&0.05&0.09\\
  $\alpha = 0.6$  & &(1.17 - 2.33)&(0.07 - 7.8)&(0 - 8.07)&(0 - 3.1)&(0 - 2.7)\\
  \cline{2-7} 
  &sh$_{Y'}$+BB$_{Y'}$&0&0.11&1.66&84.87&886.8\\
  &  &0 - 0.03&0.01 - 0.32&0.16 - 4.65&6.15 - 760.5&12.46 - 5198.1\\
  \hline
    \end{tabular}
}
\end{center}
\caption{\label{tabPL5} Synthesis of the numerical results for 4 objectives. Execution times are
  indicated in seconds. Symbol ``-'' means that the average execution time exceeds 30 minutes.}
\end{table}

\begin{table}[!h]
  \begin{center}
  {\small \begin{tabular}{|c|c|cccc|}
    \hline
      \rule[1pt]{0pt}{13pt} 
       &$n$ & 5&10&15&20\\
       \hline
     \hline
    & sh$_{Y'}$&0&0.01&0.02&-\\
       $\alpha = 0.4$   & &(0.83 - 2.4)&(0.03 - 2.27)&(0 - 1.17)& -\\
       \cline{2-6} 
       &sh$_{Y'}$+BB$_{Y'}$&0&0.45&40.74&-\\
       &  &0 - 0.02&0.01 - 2.98&1.4 - 184.8& -\\
  \hline
  \hline
   & sh$_{Y'}$&0&0.01&0.02&0.06\\
$\alpha = 0.5$  & &(0.87 - 1.93)&(0 - 1.6)&(0 - 0.1)&(0 - 0)\\

    \cline{2-6} 
 &sh$_{Y'}$+BB$_{Y'}$&0.01&0.72&47.93&1674.9\\
  &  &0 - 0.02&0.05 - 4.45&1.27 - 528.6&139 - 8378.1\\
  \hline
  \hline
  & sh$_{Y'}$&0&0.01&0.04&0.05\\
  $\alpha = 0.6$  & &(0.47 - 1.2)&(0 - 1.47)&(0 - 1.4)&(0 - 0)\\
  \cline{2-6} 
  &sh$_{Y'}$+BB$_{Y'}$&0.01&0.87&40.12&760.9\\
  &  &0 - 0.03&0.1 - 3.96&0.2 - 267& 29.61 - 3611.1\\
  \hline
    \end{tabular}
}

\end{center}
\caption{\label{tabPL6} Synthesis of the numerical results for 5 objectives. Execution times are
  indicated in seconds. Symbol ``-'' means that the average execution time exceeds 30 minutes.}
\end{table}

In these tables, one can observe that parameter $\alpha$ has a significant impact on the resolution times: it is easier for values of $\alpha$ that are above 0.5. For
example, with 3 objectives and $n = 25$, the
average execution time is 5.09 seconds for $\alpha = 0.6$, while it
is 195.6 seconds for $\alpha = 0.5$, and it is more than 30 minutes
for $\alpha = 0.4$. The sizes of the tackled instances are more modest than for strictly decreasing weights. The only available alternative method should be however to resort to an exhaustive enumeration procedure \cite{KapoR95}, that would become
intractable for 12 vertices. The results presented here are therefore
a first step towards optimizing non-convex aggregation functions in
multiobjective spanning tree problems. They make it possible to handle
instances the size of which grows up to 25 vertices.

\section{Conclusion}

In this paper we have proposed several methods to solve the
$\owa$-optimal spanning tree problem. One is based on a MIP
formulation and is valid only when the weights of the $\owa$ operator
are strictly decreasing. We have shown that the use of a preprocessing
phase, as well as a shaving procedure, makes it possible to
considerably reduce the size of the problem, and therefore speed up
the resolution. The numerical results prove the efficiency of the global resolution procedure (one is able to solve multiobjective instances with up to 100 vertices in reasonable times). Two branch and bound algorithms have also been
presented, that work whatever weights are used. Despite a greater resolution time in the strictly decreasing case, they make it possible to tackle the non-monotonic case. This is a very challenging task, and the numerical results presented here are a first step in this direction. For future works, it would be worth further investigating optimization of OWA in this case, as well as some interesting variations of OWA, namely the non-monotonic OWA operator \cite{Yager99} (where negative
weights are allowed) and the weighted OWA operator \cite{Torra97}
(where importance weights specific to each objective are allowed in
addition to the weights of the OWA operator). Other promising research
direction would be to propose a MIP formulation for a broader subclass
of Choquet integrals (enabling to take into account positive and
negative interactions between objectives), as for instance the class of
$k$-additive Choquet integrals \cite{Grabi97}. These research tracks
are especially important because the search for a single best
compromise solution is nearly the only operational approach when there
are more than three objectives and the problem does not fit into
the dynamic programming framework. 


%
%


\bibliographystyle{plain}      
\bibliography{lgos4or09}   

%
%

\end{document}